\documentclass[10pt,conference]{IEEEtran}
\usepackage[cmex10]{amsmath}
\usepackage{amssymb}
\usepackage{array}
\usepackage{url}
\newtheorem{theorem}{Theorem}[section]

\newtheorem{proposition}{Proposition}

\newtheorem{definition}[theorem]{Definition}

\newcommand{\Ac}{\ensuremath{\mathcal  A}}
\newcommand{\Bc}{\ensuremath{\mathcal  B}}
\newcommand{\Cc}{\ensuremath{\mathcal  C}}

\newcommand{\Fc}{\ensuremath{\mathcal  F}}
\newcommand{\Gc}{\ensuremath{\mathcal  G}}
\newcommand{\Pc}{\ensuremath{\mathcal  P}}
\newcommand{\Sc}{\ensuremath{\mathcal  S}}

\newcommand{\xb}{\ensuremath{\mathbf{x}}}
\newcommand{\yb}{\ensuremath{\mathbf{y}}}

\newcommand{\Pb}{\ensuremath{\mathbf{P}}}

\newcommand{\Dr}{\ensuremath{\mathrm{D}}}
\newcommand{\dr}{\ensuremath{\mathrm{d}}}
\newcommand{\Ir}{\ensuremath{\mathrm{I}}}
\newcommand{\Nr}{\ensuremath{\mathrm{N}}}

\newcommand{\bF}{\ensuremath{\mathbb{F}}}
\newcommand{\bP}{\ensuremath{\mathbb{P}}}

\newcounter{rno}
\newenvironment{rlist}{
\begin{list}{{\normalfont(\roman{rno})}}
{
\setlength{\topsep}{0.25ex}
\usecounter{rno}
\setlength{\topsep}{0.75ex}
\setlength{\labelwidth}{3ex}
\setlength{\leftmargin}{4ex}
\setlength{\labelsep}{1ex}	
\setlength{\rightmargin}{0ex}
\setlength{\itemindent}{0ex}
\setlength{\parsep}{0ex}
\setlength{\itemsep}{0.5ex plus0.2ex minus0.1ex}
}
}{\end{list}}

\newcounter{arno}
\newenvironment{arlist}{
\begin{list}{{\normalfont(\arabic{arno})\hfill}}
{
\usecounter{arno}
\setlength{\topsep}{0.75ex}
\setlength{\labelwidth}{3ex}
\setlength{\leftmargin}{4ex}
\setlength{\labelsep}{1ex}	
\setlength{\rightmargin}{0ex}
\setlength{\itemindent}{0ex}
\setlength{\parsep}{0ex}
\setlength{\itemsep}{0.5ex plus0.2ex minus0.1ex}
}
}{\end{list}}

\title{Subspace Codes for Random Networks Based on Pl\"{u}cker coordinates and Schubert Cells}
\author{\authorblockN{Anirban~Ghatak}
\authorblockA{Department of ECE\\
Indian Institute of Science, Bangalore, India-560012 \\
Email: aghatak@ece.iisc.ernet.in}}

\begin{document}

\maketitle

\begin{abstract}
The construction of codes in the projective space for error control in random networks has been the focus of recent research. The Pl\"{u}cker coordinate description of subspaces has been discussed in the context of constant dimension codes, as well as the Schubert cell description of certain code parameters. In this paper we use this classical tool to reformulate some standard constructions of constant dimension codes and give a unified framework. We present a general method of constructing non-constant dimension subspace codes with respect to minimum subspace distance or minimum injection distance as the union of constant dimension subspace codes restricted to selected Schubert cells. The selection of these Schubert cells is based on the subset distance of tuples corresponding to the Pl\"{u}cker coordinate matrices associated with the subspaces contained in the respective Schubert cells. In this context, we show that a recent construction of non-constant dimension Ferrers-diagram rank-metric subspace codes is subsumed in our framework. 
\end{abstract}
\section{Introduction}
A linear multicast network coding model is defined as follows: A source node transmits $ n $ packets, each an $ m $-ary symbol over $ {\bF}_q $; each node in the network transmits $ {\bF}_q $-linear combinations of the incoming packets at that node, and at any destination node the received packets may be represented as the rows of an $N \times m$ matrix $ Y = AX + DZ $, where the matrices $ X, Z $ are of respective dimensions $ n \times m, t \times m $ over $ {\bF}_q $. The matrix $ X $ has $ n $ source packets as its rows, $ A $ is the transfer matrix of the network, the rows of $ Z $ are the error packets, which may be random or introduced by an adversary, and $ D $ the transfer matrix with respect to these error packets. In \emph{non-coherent} or \emph{random} network coding, it is generally assumed that the source and the destination nodes have no knowledge of the network topology - the only assumption is that the number of error packets, i.e. the parameter $ t $, is bounded. \\
The motivation for construction of subspace codes for random networks over the projective space came from the seminal work of K\"{o}tter and Kschischang \cite{KK}, where the problem of error-and-erasure correction in random linear networks was translated into that of transmission and recovery of linear subspaces of a projective space. Subsequent research (\cite{SKK}, \cite{GaYan}) constructed subspace codes, with \emph{subspace distance} (\cite{KK}) as the metric, based on rank-metric codes. A very significant development in the construction of constant dimension subspace codes was the use of Ferrers diagram representation of rank-metric codes which fit into a row-reduced echelon form representation of subspaces in a Grassmannian (\cite{ES}), which was further employed in conjunction with a lexicographic ordering of subspaces (\cite{SE}). The above technique and its variants (\cite{EF2}) have resulted in some of the best-known constant dimension subspace codes till date. Other important constructions of constant dimension subspace codes include those presented in \cite{KoKu}, \cite{MGR}, \cite{Matroid}.\\
Restriction of codeword subspaces to a particular Grassmannian yields constant dimension subspace codes; a natural generalization was to attempt the construction of non-constant dimension subspace codes over the entire projective space. In this context, a new metric, namely the \emph{injection distance}, was defined in \cite{SK}, which was shown in some cases (for instance, a worst-case adversarial model) to give a more precise measure of the performance of a non-constant dimension code than the subspace distance metric. In \cite{KhK}, subspace codes (non-constant dimension) were constructed for the injection metric based on the methods in \cite{ES}. Bounds on the size of projective space codes, analogous to classical coding theoretic ones, and some explicit constructions of non-constant dimension codes were presented in \cite{EV}.\\
A new paradigm in random network coding was introduced by the so-called orbit codes \cite{Orbit}, which generalizes some of the results in \cite{KoKu} and describes constant dimension subspace codes as orbits of the action of suitable subgroups of the general linear group on a Grassmannian. Subsequent work (\cite{TRo, Cyclor1, Cyclor2}) characterized orbit codes which result from the action of cyclic subgroups. Recently, in \cite{RoTP}, the classical tool of Pl\"{u}cker coordinates was employed to compute the so-called `Pl\"{u}cker orbits' associated with an orbit code and the defining equations of Schubert varieties were used in the description of constant dimension subspace codes.\\
The contributions of this paper are listed below.
\begin{arlist}
\item We have used some classical results (\cite{Kleiman}) to provide a unified framework for some standard constructions of subspace codes (e.g. in \cite{KK}, \cite{SKK}, \cite{ES}) in terms of the Pl\"{u}cker coordinates of the constructed subspaces. In particular, we show that a strong correspondence exists: the generating matrices in both the so-called `lifting construction' and its generalization correspond exactly to the matrix of the independent Pl\"{u}cker coordinates with appropriate constraints on the choice of indices.
\item We show that the characteristic representation of a Schubert cell in a projective space in terms of the indices corresponding to a non-zero Pl\"{u}cker coordinate can be used as an ``identifying vector''(\cite{ES}, \cite{SE}, \cite{KhK}) for subspaces in that Schubert cell. In particular, we show that both the subspace distance and injection distance among the subspaces in a projective space can be bounded using the subset distance among the defining index sets of corresponding Schubert cells.
\item We give a construction of non-constant dimension subspace codes which is an extension of a construction in \cite{Orbit} to the non-constant dimension case - for both subspace distance and injection distance. The codes are described as unions of Schubert cells selected to guarantee a minimum subspace or injection distance. Moreover, we show that the Khaleghi-Kschischang construction \cite{KhK} for non-constant dimension Ferrers-diagram rank-metric subspace codes is subsumed in our proposed construction if the constant dimension subcode within each Schubert cell is a suitable rank-metric code.
\end{arlist}
The organisation of this paper is as follows: Section \ref{sec:math} is a revision of pertinent mathematical results regarding the Pl\"{u}cker coordinate description of projective subspaces and Schubert cells. In Section \ref{sec:Lift}, we reformulate standard subspace code constructions in the language of Pl\"{u}cker coordinates and give a unified framework. The next section (Section \ref{sec:Code}) gives the general construction of non-constant dimension subspace codes. We conclude with a summary of the results and a discussion on the directions of ongoing and future research in this context.
\section{Preliminary Concepts}\label{sec:math}
In this section we briefly review and organize some mathematical concepts which will be used in interpreting recent results in random network code constructions and lead to the formulation of new coding schemes.
\subsection{Projective Space over a Field}
We begin with the definition of a projective space over an arbitrary field.
\begin{definition}
\cite{Samuel} The \emph{projective space} associated with a vector space $V$ over a field $K$ is the set ${\bP}(V)$ of lines in $V$.
\end{definition}
Moreover, the projective space ${\bP}(V)$ is associated with the $K$-vector space $V$ in the following canonical way \cite{Samuel}:\\
Consider the equivalence relation $\sim$ on the set $V^* := V \setminus \{0\}$, where, for ${\xb},{\yb} \in V^{*},\, {\xb} \sim {\yb} $ if and only if ${\yb} = a{\xb},\, a \in K, \, a \neq 0 $. Then the canonical map $\Pc: V^* \longmapsto {\bP}(V)$ associates each vector $\xb$ with the projective point $K {\xb}$.\\
The dimension of ${\bP}(V)$ is defined as ${\rm{Dim}}({\bP}(V)) := {\rm{Dim}}(V) -1$; ${\bP}(0)$ is empty, and its dimension is $-1$. Therefore, a $d$-dimensional subspace or a $d$-space of the $n$-dimensional projective space ${\bP}^{n}(K)$ over a field $K$ is a set of points whose representative vectors alongwith the zero vector, form a $d+1$-dimensional subspace of the $n+1$-dimensional $K$-vector space $V$, where ${\bP}^{n}(K) = {\Pc}(V)$ \cite{Hirschfeld}. 
\subsection{The Pl\"{u}cker Coordinates of Subspaces}
A $d$-dimensional linear space $L_d$ in ${\bP}^n$ is a set of points $P = (p(0), \cdots,p(n))\,\in\,{\bP}^n$ whose coordinates $p(j)$ satisfy a system of $n-d$ independent linear equations:
\begin{equation*}
\displaystyle{ \sum^{n}_{j=0}\, \alpha_{ij}\, p(j)\, = \,0;\,\, i= 1,\cdots,n-d.}
\end{equation*}
There exist $d+1$ points $P_i = (p_i(0), \cdots,p_i(n))\,\,; i = 0,\,\cdots,\, d,\, $ which span the $d$-space $L_d$. Consider the $(d+1)\times (n+1)$ matrix $\Pb = (p_i(j))$; denote as $\Dr(j_0,\cdots,j_d)$ the determinant of the $(d+1)\times (d+1)$ submatrix formed by those columns of $\Pb$ indexed by the sequence of integers $j_0 \leq \cdots \leq j_d,\, 0\leq j_k \leq n$. It follows that at least one of the $\Nr +1$ such determinants, where $\Nr := \binom {n+1}{d+1} - 1$, must be non-zero. Written in a lexicographic order on the column indices, $(\cdots,\Dr(j_0,\cdots,j_d),\cdots )$ determines a point in ${\bP}^{\Nr}$; these determinants constitute the \emph{Pl\"{u}cker coordinates} of the $d$-dimensional linear (projective) subspace $L_d$.
\begin{theorem}\label{th:QR}
({\rm{\cite{Kleiman, HoPe}}})\, There exists a natural bijective correspondence between the $d$-spaces in $\bP^n$ and the points of ${\bP}^{\Nr}$ whose Pl\"{u}cker coordinates satisfy the following relations:
\begin{equation}\label{eq:qr}
\sum^{d+1}_{\alpha =0} (-1)^{\alpha}\Dr(j_0,\cdots,j_{d-1}, k_{\alpha})\Dr(k_0,\cdots,\check{k_{\alpha}},\cdots,k_{d+1}) = 0.
\end{equation} 
where $\check{k_{\alpha}}$ denotes omission of $ k_{\alpha} $ and $0 \leq j_{\beta},\, k_{\gamma} \leq n$.
\end{theorem}
It is easily verified that the Pl\"{u}cker coordinates of any $d$-space in $\bP^n$ indeed satisfy the above relations. Proving the converse, i.e. any point of ${\bP}^{\Nr}$ whose coordinates satisfy the above relations corresponds to a unique $d$-plane in $\bP^n$, offers some interesting insights \cite{Kleiman}. For instance, it can be shown that any $d$-subspace of the projective space $\bP^n$ has an associated matrix representation having the $(d+1)\times(d+1)$ identity matrix as a sub-matrix. The following proposition spells this out in detail.
\begin{proposition}\label{prop:id}{\rm{\cite{Kleiman}}}\,
There is a natural bijective correspondence between the set of points of $\,{\bP}^{\Nr}\,$ whose coordinates satisfy the quadratic relations in {\rm Theorem} \ref{th:QR}, with the condition $\Dr(k_0,\cdots,k_d)\neq 0$, and the affine $(d+1)(n-d)$ space of $(d+1)\times(n+1)$ matrices $(p_i(j)), \, i=0,\cdots,d,\, j =0,\cdots,n$, such that the $(d+1)\times(d+1)$ sub-matrix $p_i(k_{\alpha}),\, i=0,\cdots,d,\, \alpha =0,\cdots,d $ is the identity. Moreover, such a matrix $p_i(j)$ corresponds to the point of $\,{\bP}^{\Nr}$ with coordinates $\Dr(j_0,\cdots,j_d) = {\rm{Det}}(p_i(j_{\beta}))$; such a point of $\,{\bP}^{\Nr}$ corresponds to a $(d+1)\times(n+1)$ matrix with entries:
\begin{equation} 
p_i(j) = \Dr(k_0,\cdots,k_{i-1},j,k_{i+1},\cdots,k_d)/ \Dr(k_0,\cdots,k_d).
\end{equation}
\end{proposition}
It follows from the proposition that the set of points of ${\bP}^{\Nr}$, whose co-ordinates satisfy the quadratic relations, is covered by $N+1$ copies of the affine $(d+1)(n-d)$ space; this is the \emph{Grassmann manifold} or the \emph{Grassmannian} denoted by ${\Gc}(d+1,n+1)$. In the next section we show that the above proposition is the key to the interpretation of some well-known constructions for constant dimension codes for random networks in terms of the Pl\"{u}cker coordinates of subspaces. Henceforth we term the matrix $ (p_i(j)) $ as the \emph{Pl\"{u}cker coordinate matrix} of the subspace whose representation as a point of ${\bP}^{\Nr}$ is given by: $ (\cdots, p_i (j), \cdots) $ in lexicographic order. 
\subsection{Schubert Cells and their Representation}\label{subs:Schubert}
Let ${\Gc}(d,n)$ be the Grassmannian of $d$-dimensional subspaces of an $n$-dimensional $K$-vector space. For $U \in {\Gc}(d,n)$, there exists a uniquely determined basis $\{ v_1, v_2,\cdots,v_d \},\, v_i \in K^{n}$, of the form:\\
$ v_1 = (\ast ,\cdots , \ast , 1, 0 ,\cdots , 0)$,\\ 
$v_2 =(\ast, \cdots , \ast , 0 , \ast , \cdots ,\ast, 1, 0, \cdots, 0)$,\\
$\vdots$\\
$v_d =( \ast ,\cdots ,\ast, \underbrace{0,\ast,\cdots,\ast,0,\ast,\cdots,\ast}_{(d-1) \, \,\text{zeros}},1,0,\cdots,0)$.\\
The subspace $U$ is uniquely determined by the placement of the terminal $1$'s in the $v_i$'s, indexed by a sequence of integers $1 \leq b_1 < b_2 < \cdots < b_d \leq n $, and the elements of $K$ which occupy the $\ast$ positions. All the subspaces which are parametrized by a particular tuple $\beta = (b_1,b_2,\cdots,b_d)$, denoting the positions of the trailing $1$'s in a basis of the above form, are said to constitute the \emph{Schubert Cell} $S_{\beta}$ of dimension:
\begin{equation*}
d_{\beta} = (b_1 -1)+ \cdots + (b_d - d) = \sum^{d}_{i=1}b_i -\binom{d+1}{2}.
\end{equation*}
An alternative description of Schubert cells involves subspaces whose basis vectors form a $d \times n$ matrix in row reduced echelon form, with the defining $d$-tuple indexing the positions of the leading $1$'s instead. The Schubert cell constituted of such subspaces, say, $S_{\alpha}$ with $\alpha = (a_1,a_2,\cdots,a_d)$ and the $ a_i $ satisfy $ 1 \leq a_1 < a_2 < \cdots < a_d \leq n$, has dimension:
\begin{eqnarray*} 
d_{\alpha} & = &(n-a_1-(d-1)) + \cdots + (n-a_d)\\
 & = & nd - \sum^{d}_{i=1}a_i -\binom{d}{2}.
\end{eqnarray*}
In both the cases, the maximum dimension $n(n-d)$ is achieved for the tuple $(1,\cdots,d)$ and the minimum dimension is $0$ for the tuple $(n-d+1,\cdots,n)$. The second characterization of Schubert cells is relevant to the construction of non-constant dimension subspace codes for random networks. 

\section{A Unified Framework for Lifting Constructions using Pl\"{u}cker Coordinates}\label{sec:Lift}
In this section, we reformulate two well-known methods for constructing constant dimensional subspace codes in terms of the Pl\"{u}cker coordinates of certain subspaces, providing a unified framework for both of them.
\subsection{Overview of the Lifting Constructions}
The \emph{Lifting Construction} was first proposed in \cite{KK} to yield constant dimensional Reed-Solomon-like codes for error and erasure correction in random networks. In a subsequent work \cite{SKK}, constant dimension subspace codes were constructed by `lifting' rank-metric codes as follows.
\begin{definition}
\cite{SKK} Let $X$ be an $l \times m$ matrix over ${\bF}_q$. We define the `lifting' map as follows:
\begin{equation*}
\Ir : {\bF}^{l \times m}_q \longrightarrow {\Gc}(l,l+m);\,\, X \longmapsto {\Ir}(X) = \langle [I \,\, X] \rangle
\end{equation*}
where $I$ denotes the $l\times l$ identity matrix and $\langle M \rangle$ denotes the row-space of the matrix $M$. The subspace ${\Ir}(X)$ is called a `lifting' of the matrix $X$.\\
If $\Cc$ is a rank-metric code (e.g. \cite{Gab}), the subspace code ${\Ir}(\Cc) = \{ {\Ir}(C)\, \lvert \,C \in \Cc \}$ is called the `lifting' of $\Cc$.
\end{definition}
A generalization of the lifting construction was achieved in \cite{ES} and \cite{KhK}, where matrices in the row reduced echelon form (RREF) have been used in conjunction with Ferrers diagram representation of rank-metric codes to construct the desired subspaces. We now briefly state this construction through some relevant definitions (\cite{ES}, \cite{KhK}).
\begin{definition}
A \emph{Ferrers diagram} \cite{LW} is a graphical representation of integer partitions as an array of dots such that the $ i $-th row has exactly the same number of dots as the $ i $-th element in the partition. An $ m \times n $ Ferrers diagram has $ m $ dots in the rightmost column and $ n $ dots in the topmost row.
\end{definition}
\begin{definition}
Let $\Fc$ be a $l \times m$ Ferrers diagram. A rank-metric code $\Cc$ is called a \emph{Ferrers diagram rank-metric code} (FDRM code) associated with $\Fc$ if all the codewords are $l\times m$ matrices in which the positions of the non-zero entries match those of the dots in $\Fc$. In this case, the codeword matrices are said to \emph{fit} the Ferrers diagram $\Fc$.
\end{definition}
\begin{definition}
Let $v \in \{0,1\}^{n+1}$ be a vector of weight $k$. The profile matrix of $v$, denoted ${\rm P}_v$, is a $k\times (n+1)$ matrix in RREF such that:
\begin{rlist}
\item The leading coefficients of the rows of ${\rm P}_v$ appear in the columns indexed by ${\rm{supp}}(v)$.
\item The entries in ${\rm P}_v$ other than $0$'s and $1$'s are dots.
\end{rlist}
Extracting the dots of  ${\rm P}_v$ forms a Ferrers diagram which is denoted by ${\Fc}(v)$. 
\end{definition}
\begin{definition}
Let $v \in \{0,1\}^{n+1}$ be a vector of weight $k$ and let $X \in {\bF}^{l \times m}_q ,\, l\leq k,\, m\leq n-k+1$, be a matrix which fits ${\Fc}(v)$. Then the (generalized) lifting of $X$ (by $v$) is defined as ${\Ir}_v (X) = \langle {\rm P}_v (X) \rangle \subseteq {\Gc}(l,l+m)$, where ${\rm P}_v (X)$ is the profile matrix of $v$ with the dots replaced by the entries of $X$.\\
As before, if $ \Cc $ is a rank-metric code, then ${\Ir}_v (\Cc) = \{ {\Ir}_v (C)\, \lvert \,C \in \Cc \}$ is called the generalized `lifting' of $\Cc$, which is an FDRM subspace code associated with ${\Fc}(v)$.
\end{definition}
\subsection{Unified Framework Using Pl\"{u}cker Coordinates}
To interpret the above constructions of constant dimensional subspace codes in terms of the Pl\"{u}cker coordinates of the associated subspaces, we prove the following proposition as a consequence of Theorem \ref{th:QR} and Proposition \ref{prop:id}. 
\begin{proposition}\label{prop:matrix}
If $X \in {\bF}^{(d+1) \times (n+1)}_q$ is in row-reduced echelon form (RREF), then the Pl\"{u}cker coordinate matrix associated with the subspace $ \langle X \rangle \subseteq {\Gc}(d+1,n+1)$ coincides with $ X $, upto multiplication by a unit.
\end{proposition}
\begin{proof} 
Suppose the columns $ k_0, k_1, \cdots, k_d $ of $ X $ contain the leading elements of the rows. As $ \Dr(k_0,\cdots,k_d)\neq 0 $, by Proposition \ref{prop:id}, $ \langle X \rangle $ is identified with a point of $\,{\bP}^{\Nr}$, where $\Nr = \binom {n+1}{d+1} - 1$, corresponding to a $(d+1)\times(n+1)$ matrix with entries:
\begin{equation*}
 p_i(j) = \Dr(k_0,\cdots,k_{i-1},j,k_{i+1},\cdots,k_d)/ \Dr(k_0,\cdots,k_d), 
\end{equation*}
which is the associated Pl\"{u}cker coordinate matrix.
\end{proof}
Now we give the following interpretations of the lifting construction and its generalization in terms of the Pl\"{u}cker coordinates.\\ 
\emph{Lifting Construction}: In the setting of Proposition \ref{prop:matrix}, set $d = l-1, n=l+m-1$ and $ 0 \leq k_0 < \cdots < k_d \leq n $, with $\Dr(k_0,\cdots,k_d)\neq 0$. Let $X \in {\bF}^{l \times m}_q$, for instance, a codeword matrix of an $ l\times m $ rank-metric code over $ {\bF}_q $. Select a subspace of ${\Gc}(d+1,n+1)$ whose Pl\"{u}cker coordinates satisfy the following conditions: $k_i = i, i=0,\cdots,d$, i.e. we have $\Dr(0,\cdots,l-1) =1$ and all other coordinates are obtained from those of the form $\Dr(j_0,\cdots,j_d),\, 0 \leq j_{\beta} \leq n$, at most one of the $j_{\beta}\notin \{0,\cdots, l-1\}$, and that particular column is from the matrix $X$ in ascending order. This subspace is precisely the lifted subspace ${\Ir}(X) = \langle [I \,\, X] \rangle$.\\
\emph{Generalized Lifting Construction}: The generalized lifting construction may similarly be described as the selection of subspaces with particular constraints on their Pl\"{u}cker coordinates. Given a vector $v$ as described, set $k = d+1$ and let $k_i, \, i= 0,\cdots,d,\, 0 \leq k_i \leq n$, denote the positions of non-zero components of $v$ in ascending order. If $X \in {\bF}^{l \times m}_q$ be a matrix which fits ${\Fc}(v)$, where $ l\leq k,\, m\leq n-k+1$, (e.g. $ X $ is a codeword of a FDRM code associated with the Ferrers diagram ${\Fc}(v)$) construct the augmented matrix $A(X)\in {\bF}^{k \times n-k+1}_q$ in two steps as follows: 
\begin{arlist}
\item To each column of $X$ append $k-l$ zeros from below to form $X'\in {\bF}^{k \times m}_q $.
\item Add $\,n-k-m+1\,$ all-zero columns among the columns of $X'$ preserving the order of those columns in ${\rm P}_v$ to form $A(X)$.
\end{arlist}
Select a subspace of ${\Gc}(d+1,n+1)$ whose Pl\"{u}cker coordinates satisfy the following constraints: $\Dr(k_0,\cdots,k_d) = 1$ and the remaining coordinates are obtained from those of the form $\Dr(j_0,\cdots,j_d),\, 0 \leq j_{\beta} \leq n$, where at most one of the $j_{\beta}\notin \{k_0,\cdots,k_d\}$, and that particular column is from the augmented matrix $A(X)$ in ascending order. This subspace is the lifted subspace ${\Ir}_v (X) = \langle {\rm P}_v (X) \rangle$.\\
We now give a simple example illustrating the Pl\"{u}cker coordinate description of the generalized lifting construction.\\
\textbf{Example}
Let $v = 110101 \in \{0,1\}^6$, i.e. $n =5, k=4$. Then the profile matrix of $v$ and the associated Ferrers diagram are given by:
\begin{equation*}
{\rm P}_v =
\begin{bmatrix}
1&0 &\bullet &0 &\bullet &0 \\ 0&1&\bullet &0&\bullet &0\\0&0&0&1&\bullet &0\\ 0&0&0&0&0&1
\end{bmatrix}\,\,\,\,{\Fc}(v) =
\begin{matrix}
\bullet &\bullet\\ \bullet &\bullet\\ &\bullet
\end{matrix}
\end{equation*}
In this case, $l =3 < k =4$, but $m = 2 = n-k+1$. Now if $X \in {\bF}^{3 \times 2}_2$ is a matrix which `fits' ${\Fc}(v)$, then the lifted matrix is:
\begin{equation*}
{\Ir}_v (X) =
\begin{bmatrix}
1&0 &1 &0 &1 &0 \\ 0 &1 &1 &0 &1 &0\\0&0&0&1&1&0\\ 0&0&0&0&0&1
\end{bmatrix}
\end{equation*}
The positions of the non-zero entries in $v$ are indexed as: $k_0 =0, k_1 =1, k_2 =3, k_3 = 5$. Hence the matrix corresponding to the Pl\"{u}cker coordinate $\Dr(0,1,3,5)$ is the $4\times 4$ identity matrix and so, $\Dr(0,1,3,5) =1$. As $k-l =1$ and $m =n-k+1$, both the matrix $X'$ and the augmented matrix $A(X)$ are given by:
\begin{equation*}
X'= A(X) = \begin{bmatrix}
1&1\\1&1\\0&1 \\0 &0 
\end{bmatrix}
\end{equation*}
Consider the four coordinates: $\Dr(2,1,3,5), \Dr(0,2,3,5),$ $ \Dr(0,1,2,5), \Dr(0,1,3,2)$, involving the first column of $A(X)$ arising out of the column of ${\rm P}_v $ with index $2$. We have:
\begin{equation*}
{\rm{Det}} \begin{bmatrix}
1&0&0&0\\1&1&0&0\\0&0&1&0 \\0&0&0&1
\end{bmatrix}= 1,\, {\rm{Det}} \begin{bmatrix}
1&1&0&0\\0&1&0&0\\0&0&1&0\\0&0&0&1
\end{bmatrix}= 1
\end{equation*},
\begin{equation*}
{\rm{Det}} \begin{bmatrix}
1&0&1&0\\0&1&1&0\\0&0&0&0\\0&0&0&1
\end{bmatrix}= 0,\,{\rm{Det}} \begin{bmatrix}
1&0&0&1\\0&1&0&1\\0&0&1&0\\0&0&0&0
\end{bmatrix}= 0.
\end{equation*}
which is in accordance with Proposition \ref{prop:matrix}; and the Pl\"{u}cker coordinate matrix coincides with the lifted matrix ${\Ir}_v (X)$.\\
It follows that both the lifting construction and its generalization have precise interpretations in terms of the Pl\"{u}cker coordinate description of subspaces. In fact this interpretation provides a unified framework as the encoding procedure is shown to be equivalent to a choice of a set of column indices from the underlying (codeword) matrices and computing the Pl\"{u}cker coordinates for those particular indices. In the lifting construction the `pivotal' set of indices $\{k_0,\cdots,k_d\}$ was constrained to be just the set $\{0,\cdots, d\}$. Hence it actually limits the choice of subspaces to the so-called \emph{principal Schubert cell}, with a corresponding matrix representation having an identity matrix in the leftmost $d+1$ columns. In the generalized lifting construction, the $(d+1)\times(d+1)$ identity sub-matrix is interspersed with columns from the underlying (codeword) matrix; hence, other Schubert cells are also exploited. But, in both cases the lifted matrix still coincides with the $(d+1)\times(n+1)$ matrix formed by the `independent' Pl\"{u}cker coordinates in the setting of Theorem \ref{th:QR} and Proposition \ref{prop:id}, as has been outlined earlier.
\section{A Construction of Non-constant Dimensional Subspace Codes}\label{sec:Code}
The construction of non-constant dimensional subspace codes has been addressed in \cite{ES} and \cite{KhK}: in the former, by `puncturing' a constant dimension code that results from a multilevel code construction; in the latter, by partitioning the projective space followed by a two-level selection of subspaces. In this section we present a construction of non-constant dimension subspace codes which is an extension of the construction of orbit codes for random networks \cite{Orbit} for constant dimension codes. The Khaleghi-Kschischang construction \cite{KhK} also employs the so-called injection distance metric, rather than the subspace distance metric, as the former sometimes outperforms the latter in a worst-case adversarial model. So we first outline our construction for subspace distance and then discuss the modifications required with injection distance as the metric.
\section{Codes with Subspace Distance as Metric }\label{sec:d_S}
First we describe the selection of subspaces in terms of their Pl\"{u}cker coordinates. Recall from Proposition \ref{prop:id} that every $d$-space in the projective $n$-space has a non-zero Pl\"{u}cker coordinate $\Dr(k_0,\cdots,k_d)$ such that the $(d+1)\times (d+1)$ submatrix of the coordinate matrix formed by the columns indexed by $k_i,\, i= 0,\cdots,d$, is the identity. Let $U,\, V$ be two different subspaces, of possibly different dimensions $d_{u},\, d_{v}$, with the tuples $\{k_{u}\}$ and $\{l_{v}\}$ as the column indices of the identity submatrix in each case. Then the symmetric distance between the tuples $\{k_{u}\}$ and $\{l_{v}\}$ gives a lower bound on the subspace distance ${\dr}_s(U,V)$ as follows.
\begin{proposition}\label{prop:dsymdS}
Let $U,\, V$ be subspaces of ${{P}}^n $ of projective dimensions $s$ and $t$ respectively; and let $\{k_0, k_1,\cdots,k_{s}\}, \, 0 \leq k_{\alpha} \leq n$, and $\{l_0, l_1,\cdots,l_{t}\}, \, 0 \leq l_{\beta} \leq n$, be the tuples specifying the identity submatrices of the respective Pl\"{u}cker coordinate matrices. Then ${\dr}_s(U,V)\geq {\Delta} (\{k_{\alpha}\},\{l_{\beta}\} )$, where ${\Delta}(\{u\},\{v\})$ is the symmetric distance between the sets $\{u\},\, \{v\}$.
\end{proposition}
\begin{proof}
Without loss of generality assume that the Pl\"{u}cker coordinate matrices $M(U),\, M(V)$, of respective dimensions $(s+1)\times(n+1)$ and $(t+1)\times(n+1)$, are in RREF, keeping the columns of the identity submatrices intact. Now 
\begin{equation*}
{\Delta} (\{k_{\alpha}\},\{l_{\beta}\} ) = \lvert  \{k_{\alpha}\}\setminus\{l_{\beta}\}\rvert + \lvert \{l_{\beta}\}\setminus \{k_{\alpha}\}\rvert.
\end{equation*}
Let $a:=\lvert \{k_{\alpha}\}\setminus\{l_{\beta}\}\rvert,\, b:= \lvert \{l_{\beta}\}\setminus \{k_{\alpha}\}\rvert$. Then there are $s +1-a =t+1-b$ matching positions in the matrices $M(U),\, M(V)$ where there are leading 1's; hence we have ${\rm{Dim}}(U \cap V) \leq (1/2)(s+t -(a+b))$. But then 
\begin{eqnarray*}
{\dr}_s(U,V)& :=&{\rm{Dim}}\,U+{\rm{Dim}}\,V - 2{\rm{Dim}}(U \cap V)\\
            &\geq & s+t -2((1/2)(s+t -(a+b)))\\
            &= & a+b.
\end{eqnarray*}
\end{proof}
The above result is a general reformulation of Lemma $2$ in \cite{ES}, which relates the subspace distance between the subspaces of a projective space with the Hamming distance between their binary profile vectors. \\
Following \cite{KhK}, a two-step strategy for constructing non-constant dimension subspace codes over a projective space is outlined below:
\begin{arlist}
\item Selection of Schubert cells whose constituent subspaces across the cells are at a prescribed minimum subspace distance from each other;
\item Selection of subsets of subspaces in each Schubert cell which have the prescribed minimum subspace distance among them.
\end{arlist}
In \cite{KhK}, the selection of subsets within a Schubert cell is achieved via lifting of Ferrers diagram rank-metric (FDRM) codes, while a profile vector selection algorithm is implemented to choose appropriate Schubert cells which also support the largest lifted FDRM codes. As discussed in preceding sections, Schubert cells can be characterized by the tuples which index the positions of leading 1's of a matrix in RREF. Following Proposition \ref{prop:dsymdS}, the selection of a Schubert cell is translated into the selection of a tuple of column indices in the Pl\"{u}cker coordinate matrix of a representative subspace. For the selection of subspaces within a Schubert cell we propose a method exploiting the characteristic representation of subspaces in the Schubert cell. First we state some relevant definitions.
\begin{definition}
Let the set of basis vectors of any subspace of a Schubert cell $S_{\alpha}$ have a $d \times n$ matrix representation in RREF with the positions of the leading $1$'s indexed by the tuple $\{\alpha\} =\{a_1, a_2, \cdots, a_d\}$. The \emph{Schubert cell matrix} $A_{\alpha}$ corresponding to the Schubert cell $S_{\alpha}$ is the $d \times n$ matrix in RREF in which the columns indexed by the tuple $\{\alpha\}$ form a $d \times d$ identity matrix and all other columns are zero.
\end{definition} 
\begin{definition}
Given a Schubert cell $S_{\alpha}$ characterized by the tuple $\{\alpha \} =\{a_1, a_2, $ $\cdots,a_d\}$, $1\leq a_1 < \cdots < a_d \leq n$, the \emph{complementary Schubert cell} $S_{\alpha}^{c}$ is characterized by the tuple $ \{ \beta \} = \{ b_1,\cdots,b_{(n-d)}\} = \{1,2,\cdots,n\}\setminus \{a_1, a_2, \cdots, a_d\}$.
\end{definition}
Evidently the Schubert cell matrix for $S_{\alpha}^{c}$ is given by the $(n-d) \times n$ matrix $A_{\alpha}^{c}$ in RREF with the columns indexed by $\{\beta\}$ forming an $(n-d)\times (n-d)$ identity matrix and the remaining columns all zero.\\
Let $S_{\alpha}$ be a Schubert cell with the associated $d \times n$ Schubert cell matrix $A_{\alpha}$. Also denote the $d\times n$ RREF `matrix' representation of the $S_{\alpha}$ with asterisks indicating arbitrary field elements (vide Subsection \ref{subs:Schubert}) as $M(S_{\alpha}).$\\
Consider the set of $d \times (n-d)$ matrices such that the columns of these matrices correspond to those columns indexed by the set: $\beta = \{b_1, b_2, \cdots, b_{n-d}\}:= \{1,2,\cdots, n \} \setminus \{a_1, a_2, \cdots, a_d \}$ in ascending order, having non-zero elements only in the positions of the asterisks in $M(S_{\alpha})$. Select an additive group $\Gc$ of these matrices such that each element of the group has rank at least $r/2$.\\
For each $G \in \Gc$ construct the $d \times n$ matrix $C(G)$ as follows: 
\begin{equation*}
C(G) = A_{\alpha} + G \cdot A_{\alpha}^{c}.
\end{equation*}
The constant dimension subspace code constructed within the Schubert cell $S_{\alpha}$ is then given by: \begin{equation*}
{\Cc}_{\alpha} = \{\langle\,C(G)\,\rangle \lvert \, G \in \Gc \}.
\end{equation*}
The minimum subspace distance of the code satisfies ${\dr}_{min} \geq r$, as borne out in the following theorem.
\begin{theorem}
Let $G_1,\, G_2$ belong to an additive group of $d \times (n-d)$ matrices with minimum rank $r/2$. Then the subspaces spanned by the matrices $C(G_1)$ and $C(G_2)$ within the Schubert cell $S_{\alpha}$ are at a minimum subspace distance ${\dr}_{min} \geq r$, where $C(G_i), \, i= 1,2$, are defined as above.
\end{theorem}
\begin{proof} The $d \times n$ matrices $C_i = C(G_i), \, i= 1,2 $, have the matrices $G_1, \,G_2$, which belong to the additive group and have minimum rank $r/2$, as the $d\times (n-d)$ submatrices constituted of the columns indexed by $b_1, b_2, \cdots, b_{n-d}$. Then we have:
\begin{equation*}
{\dr}_{S}(C_1,\,C_2) = 2\,{\text{rank}}\,\begin{bmatrix} C_1\\C_2 \end{bmatrix} - 2d = 2\,{\text{rank}}\,\begin{bmatrix} C_1-C_2 \\C_2 \end{bmatrix} - 2d \geq r.
\end{equation*}
The last step follows from the fact that the non-zero columns of the difference matrix $C_1-C_2$ constitute an element of the additive group of matrices of minimum rank $r/2$. 
\end{proof}
Hence we can describe our two-step construction of non-constant dimension subspace codes with a minimum subspace distance ${\dr}_{min}$ as follows.
\begin{arlist}
\item{\emph{Selection of Schubert cells}: Select the largest subset $\Sc$ of Schubert cells subject to the following condition: For any $S_{\alpha}, S_{\beta} \in \Sc$, we have ${\Delta} (\{\alpha\},\{\beta\}) \geq {\dr}_{min}$;}
\item{\emph{Selection of subspaces in each Schubert cell}: The subspaces within a Schubert cell are chosen such that their spanning matrices satisfy the following condition: The submatrix indexed by the characteristic columns of the complementary Schubert cell matrix belongs to an additive group of minimum rank $\lfloor {\dr}_{min}/2 \rfloor$.} 
\end{arlist}
The non-constant dimension subspace code is, therefore, the union: $\bigcup {\Cc}_{\alpha}$, where the union is over all the Schubert cells $S_{\alpha} \in \Sc$.
\section{Codes with Injection Distance as Metric}\label{sec:d_I}
In this section we first introduce the notion of injection distance and discuss its applicability vis-$\grave{a}$-vis subspace distance. Next we modify the construction of non-constant dimension subspace codes presented earlier to adapt it for the injection distance metric.
\subsection{The Adversarial Channel Model and Injection Distance between Subspaces:}
Recall that in the linear multicast network coding model the received packets at a destination are expressed as the rows of a matrix:  $ Y = AX + DZ $, where the matrix $ A $ is the network transfer matrix seen by the source packets (rows of $ X $) and $ D $ is the transfer matrix seen by the error packets (rows of $ Z $). In the general random network coding problem, the matrix $ A $ is unknown to the receiver except for a lower bound on the rank of the matrix. In the worst-case adversarial model investigated in \cite{SK}, the following additional assumptions are made:
\begin{arlist}
\item The adversary knows the matrix $ X $ and can choose the transfer matrix $ A $ while respecting the rank constraint;
\item The matrix $ D $ is arbitrarily chosen by the adversary.
\end{arlist}
\begin{definition}\cite{SK}
The \emph{injection distance} between subspaces $ U, \, V $ of ${\bP}^n$ is given by:
\begin{eqnarray*}
{\dr}_I (U,V) := {\rm{max}}\{{\rm{Dim}}\,U,{\rm{Dim}}\,V \} - {\rm{Dim}}(U \cap V) \\
 = {\rm{Dim}}(U + V)- {\rm{min}}\{{\rm{Dim}}\,U,{\rm{Dim}}\,V \}.
 \end{eqnarray*}
\end{definition}
The injection distance between two subspaces is thus related to the subspace distance between them as follows (\cite{SK}):
\begin{equation}\label{eqn:dSdI}
{\dr}_I (U,V) = \frac{1}{2} {\dr}_s(U,V) + \frac{1}{2} \lvert {\rm{Dim}}\,U - {\rm{Dim}}\,V \rvert 
\end{equation}
It is evident that, upto a scaling factor, the injection distance coincides with the subspace distance for constant-dimension subspace codes. It was observed in \cite{SK} that the injection distance gives a more precise measure of the subspace code performance when a single error packet can cause simultaneous error and erasure. However, in a less pessimistic channel model than the worst-case adversarial channel, for instance, in the case where a single error packet can introduce either an error or an erasure but not both, the subspace distance is still appropriate.
\subsection{Construction of Non-constant Dimension Subspace Codes:}
First we prove a counterpart of Proposition \ref{prop:dsymdS}, with the same setting and notation, for the injection distance ${\dr}_I(U,V)$ between subspaces $U,\, V$ of ${\bP}^n$.
\begin{proposition}\label{prop:dsymdI}
Let $U,\, V$ be subspaces of ${\bP}^n$ of projective dimensions $s$ and $t$ respectively; and let $\{k_0, k_1,\cdots,k_{s}\}, \, 0 \leq k_{\alpha} \leq n$, and $\{l_0, l_1,\cdots,l_{t}\}, \, 0 \leq l_{\beta} \leq n$, be the tuples specifying the identity submatrices of the respective Pl\"{u}cker coordinate matrices. Then ${\dr}_I(U,V) \geq \lfloor \frac{1}{2} {\Delta} (\{k_{\alpha}\},\{l_{\beta}\}) \rfloor$, where ${\Delta}(\{u\},\{v\})$ is the symmetric distance between the sets $\{u\},\, \{v\}$.
\end{proposition}
\begin{proof}
From the proof of Proposition \ref{prop:dsymdS}, we obtain:\\
${\rm{Dim}}(U \cap V) \leq (1/2)(s+t -(a+b))$, where $a:=\lvert  \{k_{\alpha}\} \setminus\{l_{\beta}\}\rvert$ and $ b:= \lvert \{l_{\beta}\}\setminus \{k_{\alpha}\}\rvert$. 
So we have: 
\begin{eqnarray*}
2{\dr}_I(U,V)& := & 2{\rm{max}}\{{\rm{Dim}}\,U,{\rm{Dim}}\,V \} - 2{\rm{Dim}}(U \cap V)\\
             & \geq & {\rm{Dim}}\,U + {\rm{Dim}}\,V- 2{\rm{Dim}}(U \cap V)\\
             & \geq & a+b.
\end{eqnarray*}
As ${\Delta} (\{k_{\alpha}\},\{l_{\beta}\}) = a+b $, the proposition follows.\\
Alternatively, from Equation \ref{eqn:dSdI} and Proposition \ref{prop:dsymdS}, we have:
\begin{equation*}
2{\dr}_I(U,V) \geq  {\Delta} (\{k_{\alpha}\},\{l_{\beta}\}) + \lvert s-t \rvert 
\end{equation*} 
and the proposition follows.
\end{proof}

As suggested above, to capture the effect of the injection distance metric in our framework, we modify the definition of the symmetric distance between two sets.
\begin{definition}
For two finite sets $ S_1, S_2 $ of respective cardinalities $ s_1, s_2 $, the \emph{modified symmetric distance} is defined as follows:
\begin{equation*}
\Delta_m(S_1, S_2) =  \# \{ S_1 \setminus S_2 \} + \# \{ S_2 \setminus S_1 \} + \lvert s_1 -s_2 \rvert  
\end{equation*}
where $ \# S $ denotes the cardinality of the set $ S $ and `$ \lvert \, \rvert $' denotes the absolute value.
\end{definition}
\begin{proposition}\label{prop:modsym}
The modified symmetric distance is a metric.
\end{proposition}
\begin{proof}
It is evident that $ \Delta_m(S_1, S_2) = \Delta(S_1, S_2) + \lvert s_1 - s_2 \rvert $. Symmetry and the fact that $ \Delta_m (S_1, S_2) \geq 0$ with equality if and only if $ S_1 = S_2 $ follow immediately. Triangle inequality is satisfied because $ \Delta(\,) $ is a metric and for positive integers $ s_1, s_2, s_3 $ we have: $ \lvert s_1 - s_2 \rvert + \lvert s_1 - s_3 \rvert \geq \lvert s_2 - s_3 \rvert.$
\end{proof}
Now we reformulate the two-step construction of non-constant dimension subspace codes with injection distance as the metric. From Propositions \ref{prop:dsymdI} and \ref{prop:modsym}, the selection criterion of Schubert cells for the injection distance metric is readily obtained from the subspace distance case. Moreover, as the second step is essentially a construction of a constant-dimension subspace code within a Schubert cell, the transition from subspace distance to injection distance involves only a scaling factor. Hence we may state the construction of non-constant dimension subspace codes with a minimum injection distance $ \delta_{min} $ as follows. 
\begin{arlist}
\item{\emph{Selection of Schubert cells}: Select the largest subset $\Sc$ of Schubert cells subject to the following condition: For any $S_{\alpha}, S_{\beta} \in \Sc$, ${\Delta}_m (\{\alpha\},\{\beta\}) \geq 2{\delta}_{min}$;}
\item{\emph{Selection of subspaces in each Schubert cell}: The subspaces within a Schubert cell are chosen such that their spanning matrices satisfy the following condition: The submatrix indexed by the characteristic columns of the complementary Schubert cell matrix is an element of an additive group of minimum rank ${\delta}_{min}$.} 
\end{arlist}
The non-constant dimension subspace code is the union: $\bigcup {\Cc}_{\alpha}$, over all the Schubert cells $S_{\alpha} \in \Sc$, as defined earlier.
\section{Non-constant Dimension Ferrers-diagram Rank-metric Subspace Codes}\label{sec:ncfdrm} 
 In the proposed construction of non-constant dimension subspace codes, if the additive group $ \Gc $ of $d \times (n-d)$ matrices is chosen to be a rank-metric code \cite{Gab}, we have a so-called non-constant dimension Ferrers-diagram rank-metric code (FDRM) subspace code. We show that the Khaleghi-Kschischang construction \cite{KhK}, which has resulted in some of the best-known subspace codes in terms of code rate, is subsumed in our framework when the additive group $ \Gc $ is a rank-metric code.\\
In the proposed construction, the FDRM codeword matrices have non-zero elements in positions which match the asterisks in the `matrix' $M(S_{\alpha})$, where $ S_{\alpha} $ is the Schubert cell characterized by the tuple $   \{\alpha \} =\{a_1, a_2, $ $\cdots,a_d\} $. First we prove a lower bound on the dimension of such an FDRM code, similar to Theorem $ 3 $ in \cite{KhK}, in this general setting.
\begin{theorem}\label{th:FDRM}
The dimension $ \kappa $ of the largest FDRM code whose codewords have non-zero elements precisely in the positions of the asterisks in $ M(S_{\alpha}) $ satisfies: $ \kappa \geq D_{\alpha} - {\rm{max}}\{ d, n-d \}(\delta_r -1) $, where $ \delta_r $ is the minimum rank distance of the rank-metric code and $ D_{\alpha} $, the number of asterisks in $ M(S_{\alpha}) $, is the dimension of the Schubert cell $ S_{\alpha} $.
\end{theorem} 
\begin{proof}
It is known \cite{Gab} that the dimension of a linear maximum rank-distance (MRD) code $ \Cc $, whose codewords are  $d \times (n-d)$ matrices over $ {\bF}_q $, is given by:
\begin{equation*}
 D = {\rm{max}}\{ d, n-d \}({\rm{min}}\{ d, n-d \} - {\delta}_r + 1) 
 \end{equation*}
where $ {\delta}_r $ is the minimum rank distance of $ \Cc $.\\
Define $ V := {{\bF}_q}^{d \times (n-d)} $. The kernel of the surjective map: $ {\pi} : V \rightarrow  V/ {\Cc} $ is a subspace of dimension $ d(n-d) - D $. The largest FDRM code which matches $ M(S_{\alpha}) $ is the largest subcode $ {\Cc}' \subseteq \Cc $ subject to the constraint that all its codeword matrices have zeros in $ d(n-d) - D_{\alpha} $ fixed positions. Invoking the standard vector-space isomorphism between $ {{\bF}_q}^{m \times n} $ and $ {{\bF}_q}^{mn} $ where $ m,n $ are positive integers, any codeword $ c \in {\Cc}' $ may be expressed as a $ d(n-d) $-dimensional vector over $ {\bF}_q $ with non-zero elements in $ D_{\alpha} $ coordinates $ i_1,\cdots,i_{D_{\alpha}} $. Let $\phi: V \rightarrow {{\bF}_q}^{d(n-d) - D_{\alpha}}$ be the surjective map which sets zero the coordinates $ i_1,\cdots,i_{D_{\alpha}} $. Now consider the linear map: 
\begin{equation*}
\Phi: V \longrightarrow V/{\Cc} \times {{{\bF}_q}^{d(n-d) - D_{\alpha}}};\,\, v \longmapsto (\pi(v), \phi (v)). 
\end{equation*}
Evidently, $ {\rm{Dim}}\, \Phi(V) \leq 2d(n-d) - D - D_{\alpha} $. The largest FDRM code is the kernel of the map $ \Phi $, and hence, by the rank-nullity theorem, has dimension:
\[
\begin{array}{l}
 \kappa \geq d(n-d) - (2d(n-d) - D - D_{\alpha}) \\
= D_{\alpha} + {\rm{max}} \lbrace d, n-d \rbrace ({\rm{min}} \lbrace d, n-d \rbrace - \delta_r + 1) - d(n-d)\\ 
= D_{\alpha} -  {\rm{max}}\lbrace d, n-d \rbrace(\delta_r -1).
\end{array}
\]     
\end{proof}
\subsection{Selection of Schubert Cells:} It follows from Theorem \ref{th:FDRM} that the size of the FDRM non-constant dimension subspace codes depends directly on the dimension of the chosen Schubert cells. In \cite{KhK} the relation between the so-called profile vectors associated with each Schubert cell and the dimension of the FDRM codes constructed has been used in an algorithm for selection of the Schubert cells. We present a more direct approach, basing our choice on the dimension of each Schubert cell, as follows.\\
Recall from Subsection \ref{subs:Schubert} that the dimension of a Schubert cell $ S_{\alpha} $ in the second representation characterized by the tuple $ \{\alpha \}$ is given by: $ D_{\alpha} = nd(\alpha) - \sum^{d(\alpha)}_{i=1}a_i -\binom{d(\alpha)}{2} $, where $ d(\alpha) := \lvert \{ \alpha \} \rvert = d $ is the length of the tuple. Let $ \Ac $ denote the set of all such distinct tuples $ \alpha $ in lexicographic order. For each $ \alpha \in \Ac $ define the choice function:
\begin{equation}\label{eq:choice}
F(S_{\alpha}) = D_{\alpha} - {\rm{max}}\{ d(\alpha), n-d(\alpha) \}(\delta_r -1). 
\end{equation}
The selection of Schubert cells proceeds via a greedy algorithm, similar to that in \cite{KhK}, which selects at each stage the tuple $\{ \alpha \}$ with the largest $ F(S_{\alpha}) $ and discards all the tuples $ \beta \in \Ac $ such that $ {\Delta} (\{\alpha\},\{\beta\}) < {\dr}_{min} $ for the subspace distance metric or $ \lfloor \frac{1}{2}{\Delta}(\{\alpha\},\{\beta\})\rfloor  + \lfloor \frac{1}{2} \lvert \# \{\alpha\}- \# \{\beta\} \rvert\rfloor < {\delta}_{min} $ for the injection distance metric. The next stage proceeds by applying the same choice function to the modified set $ \Ac := \Ac \setminus \Bc $, where $ \Bc $ is the set of the discarded tuples $ \{ \beta \} $ at the previous stage. The selection algorithm terminates when $ \Ac := \emptyset $. The choice function in Equation \ref{eq:choice}, which is based on the lower bound on the dimension of the FDRM subspace code given in Theorem \ref{th:FDRM}, ensures that at every stage a Schubert cell which supports the largest FDRM code is chosen among the contending cells. The largest number of subspaces belong to the principal Schubert cell of the Grassmannian ${\Gc}(d,n)$, with $ d = \lfloor n/2 \rfloor  $ or $ d = \lceil n/2 \rceil $; for these values of $ d $, the negative term in the choice function is also minimum for a given $ \delta_r $. Hence the algorithm proceeds by choosing either of these subspaces (different when $ n $ odd) in the first step.
\subsection{Comparison with the Khaleghi-Kschischang algorithm:} The above algorithm and that stated in \cite{KhK} differ in two aspects, one of which can be explained as a general reformulation and the other, a slight modification.
\begin{rlist}
\item \emph{Binary Profile Vectors vs. Characteristic Tuples:} In \cite{KhK}, the selection of Schubert cells is based on the so-called binary profile vectors associated with each cell. The Ferrers diagram representation of each profile vector then gives the dimension and the positions of non-zero entries of the appropriate rank-metric code. In our construction we simply use the characteristic tuple representation of each Schubert cell and the rank-metric code is likewise specified by the so-called Schubert cell ``matrix'' (with the asterisks in the ``basis  vectors'' replacing the dots in the Ferrers diagrams) - a direct formulation.
\item \emph{The Choice Function for the Greedy Algorithm:} The following scoring function is used to select a new profile vector $ v $ in \cite{KhK}:
\begin{equation}\label{eq:scorefk}
{\rm{score}}(v, \delta_r) = \displaystyle{\sum_{i=1}^{n} \sum_{j =1}^{i} \bar{v_i}v_j - {\rm{max}}\{wt(v), \eta(v)\}(\delta_r -1)},
\end{equation}
where $ \eta(v) = n - (wt(v) + {\rm{min}}_{t \in s(v)} t) +1 $, $ wt(v) $ the (Hamming) weight of the binary vector $ v \in \{ 0, 1 \} ^n $, $ \bar{v} $ the complement of $ v $ and $ s(v) $ the support of $ v $, i.e. the set whose elements denote the positions of the 1's in $ v $. It is evident that the sum term in the above scoring function computes $ D_{\alpha} $, the number of asterisks in the `matrix' representation of the chosen Schubert cell in our framework, which, in our choice function (Equation \ref{eq:choice}), is the exact formula based on the characteristic tuple. Rewriting the second term in Equation \ref{eq:choice} as $ {\rm{max}}\{ d(\alpha), n - (d(\alpha)+ {\rm{min}}_{i} \{ a_i \}) + 1 \}(\delta_r -1) $ makes it equivalent to the second term of Equation \ref{eq:scorefk}. It modifies the parameter $ n-d(\alpha) $ by subtracting the number of leftmost all-zero columns in $ M(S_{\alpha}) $, which is not necessary for making a choice in our framework.
\end{rlist}
\subsection{Numerical Results:} We compute lower bounds on the rates of codes constructed for a given minimum subspace distance (using symmetric distance) and for a given minimum injection distance (using the modified symmetric distance), from the values of the choice function (Equation \ref{eq:choice}) for each selected Schubert cell. For a selected set of Schubert cells $ \Sc $, a lower bound on the code rate of a non-constant dimension FDRM subspace code over a finite field $ {\bF}_q $ is given by the formula: $  \log_q (\sum_{S_{\alpha} \in \Sc } q^{\Fc(S_{\alpha})}) $ where\\
\[
\Fc(S_{\alpha}) =  
\begin{cases} 
  F(S_{\alpha}), & \text{if}\,\, F(S_{\alpha}) \geq 0\\
   0 & \text{otherwise}. 
\end{cases}
\] 
 
It is seen that the calculated lower bound is not simply a computation of the choice function in each case and then summing the values in the exponent. The justification for this modification is as follows: The choice function $ F(S_{\alpha}) $ for a Schubert cell $ S_{\alpha} \in \Sc $ is zero or negative when it does not support a rank-metric code of the given dimensions with the prescribed minimum rank-distance. However, the particular Schubert cell still maintains the inter-cell minimum subspace or injection distance and so contributes at least one subspace codeword to the entire code. To give a simple illustration, consider the following construction of a non-constant dimension FDRM subspace code in $ {\bP}^4 ({\bF}_2) $ with minimum subspace distance $ 4 $. A set $ \Sc $ for this code consists of Schubert cells specified by the tuples $ \{1, 2 \}, \{ 3, 4 \}, \{ 1, 2, 3, 4, 5 \} $, with corresponding choice function values of $ 3, -1, -5 $, respectively. For the second tuple $ M(S_{\alpha}) $ has only $ 2 $ asterisks while the last one has none, being the $ 5 \times 5 $ identity matrix, obviously not enough to support the required rank-metric code in each case. However, they can still contribute a single subspace codeword each.\\
In Table $ 1 $ we have given a few lower bounds on the code rates in our construction with respect to given minimum subspace distance (column LB($ {\dr}_s $)) and minimum injection distance (column LB($ {\dr}_I $)) for some small field sizes $ q $ and projective dimensions $ N-1 $. It is seen that they are comparable with those of the actual codes constructed in \cite{KhK}, which are also presented in the columns ` KK($ {\dr}_I $)' (minimum injection distance) and `KK($ {\dr}_s $)' (minimum subspace distance). In fact, for subspace distance the rates of the constructed codes actually coincide with the computed lower bound in half of the cases, while the other values show only marginal improvement. The difference between the rates of the constructed codes and the computed lower bounds is more pronounced for injection distance. 

\begin{table}
\begin{tabular}{|l|l|l|l|l|l|l|l|l|}
\hline
$q$ & $ {\dr}_I $ & ${\dr}_s $  & $ N $ & LB($ {\dr}_s $)  & KK($ {\dr}_s $) &  LB($ {\dr}_I $) & KK($ {\dr}_I $)  ) \\
\hline
$ 2 $ & $ 2 $ & $ 4 $ & $ 9 $ & $15.1515   $ & $ 15.1518 $ & $ 15.3238 $ & $ 15.6245 $\\
$ 2 $ & $ 2 $ & $ 4 $ & $ 10 $ & $ 20.1534 $ & $ 20.1534 $ & $ 20.1967 $ & $20.3294  $\\
$ 2 $ & $ 2 $ & $ 4 $ & $ 12 $ & $ 30.1556 $ & $ 30.1557 $ & $ 30.1998 $ & $ 30.3346 $\\
$ 2 $ & $ 3 $ & $ 6 $ & $ 13 $ & $28.0030  $ & $ 28.0030 $ & $ 28.0134 $ & $ 28.0263 $\\
$ 3 $ & $ 2 $ & $ 4 $ & $ 7 $ & $ 8.0131 $ & $  8.0145 $ & $ 8.0464 $ & $ 8.1331 $\\
$ 3 $ & $ 2 $ & $ 4 $ & $ 8 $ & $ 12.0135 $ & $ 12.0135 $ & $ 12.0160 $ & $  12.0311$\\
$ 4 $ & $ 2 $ & $ 4 $ & $ 7 $ & $ 8.0030 $ & $  8.0032 $ & $ 8.0142 $ & $8.0522 $\\
$ 4 $ & $ 2 $ & $ 4 $ & $ 8 $ & $  12.0031 $ & $ 12.0031 $ & $ 12.0034 $ & $ 12.0068$\\
\hline
\end{tabular}
\vspace{1.5 mm}
\caption{Comparison of lower bounds on the rates of our FDRM subspace codes with the actual rates of codes in \cite{KhK}}
\end{table}
\section{Conclusion}
We have established a unified framework for the construction of constant-dimension random network codes in terms of the classical description of the subspaces of a projective space using the Pl\"{u}cker coordinates. Further we have given a general construction of non-constant dimension subspace codes, which exploits a basic representation of Schubert systems, for both the subspace distance and the injection distance metrics. We have demonstrated that our framework subsumes the construction of \cite{KhK}, which is corroborated by the proximity of the estimated lower bounds for our construction with the actual code rates given for small parameters in \cite{KhK}. Some questions regarding the above constructions remain open: 
\begin{rlist}
\item For the non-constant dimension FDRM subspace codes, will an exhaustive search algorithm yield a significantly better code-rate than the lower bound based on a greedy search? Is there a systematic search procedure to achieve the best possible code-rate?
\item Is it possible to obtain a set (or better still, a group) of subspaces within a selected Schubert cell larger than the largest Ferrers-diagram rank-metric code, which satisfies the minimum subspace/injection distance condition? This problem is a constant-dimension subspace code construction problem with the added constraint of the Schubert cell structure. 
\end{rlist} 
So far there have been several approaches to the decoding problem of constant dimension codes in general (\cite{SKK}), of spread codes in particular, (\cite{GMR}) and employing list decoding methods (\cite{List1}, \cite{List2}, \cite{List3}). Decoding of `punctured' non-constant dimension codes is discussed in \cite{ES}. A possible decoding strategy of our non-constant dimension subspace codes and those in \cite{KhK} involves two steps: 
\begin{arlist}
\item Estimating the Schubert cell to which the transmitted subspace belongs. In \cite{KhK}, so-called binary asymmetric codes with minimum asymmetric distance greater than or equal to the minimum injection distance and binary codes with minimum Hamming distance greater than or equal to the minimum subspace distance have been used to estimate the profile vector corresponding to the Schubert cell containing the transmitted subspace. In our construction, a similar procedure, identifying each characteristic tuple with a codeword of a binary asymmetric code or that of a binary code with the same relations on the injection and subspace distances, would suffice for estimating the tuple corresponding to the Schubert cell. 
\item Estimating the transmitted subset within the Schubert cell, which is the decoding problem for rank-metric constant-dimension subspace codes.
\end{arlist}  
Hence another fruitful direction of research will be to devise an efficient general decoding algorithm for non-constant dimension projective space codes.
\section*{Acknowledgment}
The author is grateful to Dilip P. Patil, Sumanta Mukherjee and Smarajit Das for their assistance.

\end{document}